\documentclass{article}
\usepackage{graphicx} %
\usepackage{mathrsfs}
\usepackage{enumitem}
\usepackage{amsthm}
\usepackage[super]{nth}
\usepackage[toc,page]{appendix}
\usepackage{listings}
\usepackage{mathtools}
\usepackage{color}
\usepackage{tabularx}
\usepackage{amsmath}
\usepackage{csquotes}
\usepackage{multicol}
\usepackage{amssymb}
\usepackage{fancyvrb}
\usepackage{makecell}
\usepackage{mdframed}
\usepackage[toc,page]{appendix}
\usepackage{mathdots}
\usepackage{xspace}
\usepackage{empheq}
\usepackage{arydshln}
\usepackage{nicefrac}
\usepackage{dsfont}
\usepackage{natbib}
\usepackage[hidelinks]{hyperref}
\usepackage{cleveref}
\usepackage[ruled, linesnumbered]{algorithm2e} 

\usepackage{array}
\newcolumntype{L}[1]{>{\raggedright\let\newline\\\arraybackslash\hspace{0pt}}m{#1}}
\newcolumntype{C}[1]{>{\centering\let\newline\\\arraybackslash\hspace{0pt}}m{#1}}
\newcolumntype{R}[1]{>{\raggedleft\let\newline\\\arraybackslash\hspace{0pt}}m{#1}}

\usepackage{tikz}
\usetikzlibrary{automata, positioning, calc, shapes, arrows, fit}

\DeclareMathOperator{\sw}{sw}
\DeclareMathOperator{\dist}{dist}
\DeclareMathOperator{\plu}{plu}
\DeclareMathOperator{\boto}{bot}
\theoremstyle{plain}
\newtheorem{theorem}{Theorem}
\newtheorem{lemma}[theorem]{Lemma}

\newtheorem{definition}{Definition}

\title{A Note on Rules Achieving Optimal Metric Distortion}
\author{Jannik Peters \footnote{jannik.peters@tu-berlin.de} }
\date{\today}

\begin{document}

\maketitle

\begin{abstract}
    In this note, we uncover three connections between the metric distortion problem and voting methods and axioms from the social choice literature. 
\end{abstract} 
\section{Introduction}
The distortion of voting rules was introduced by \citet{PrRo06a} as a way to implicitly measure the utilitarian performance of a voting rule. The distortion of a voting rule given some ordinal preference profile, is the worst-case approximation factor to the social welfare this voting rule can achieve, among all utility profiles consistent with the ordinal preferences. Without any additional assumptions on these utility profiles, no voting rule can achieve any constant factor guarantees. This motivated \citet{anshelevich2018approximating} to study the special case of metric utility profiles, i.e., utility profiles for which a metric space exists, in which the voters and candidates lie, such that the utility of a voter for a candidate is precisely the distance between them.  For such utility profiles, a lower bound of $3$ on the distortion of any rule was quickly established, while it was also shown that the Copeland rule (in fact, any rule selecting from the uncovered set) achieves a distortion of $5$. This upper bound of $5$ was subsequently lowered by \citet{kempe2020analysis} and \citet{munagala2019improved}, until finally \citet{GHS20a} achieved an upper bound of $3$ through a rather ``complex'' rule dubbed Plurality Matching. This rule was subsequently simplified by \citet{KiKe22a} to the so-called Plurality Veto rule.

In this paper, we show that the Plurality Matching and Plurality Veto rules are related closely to three different settings in social choice: (i) The Proportional Veto Core by \citet{Mou81a}, (ii) the concept of Proportionality for Solid Coalitions from multiwinner voting (see, for instance, \citet{AzLe20a} or \citet{BrPe23a}), and (iii) the Matching under Preferences setting \citep{BoMo01a}. In fact, we give an equivalent formulation of the Plurality Matching concept in all three settings. Using this equivalence, the Plurality Veto rule, the Vote-by-Veto rule \citep{Mou81a}, and the Serial Dictatorship rule \citep{AbSo98a} are instantiations of the same rule. Further, Phragmén's Ordered Rule \citep{Jans16a}, Probabilistic Serial \citep{BoMo01a}, and the ``Veto by consumption'' rule \citep{IaKo21a} can also be treated as instantiations of the same rule and thus lead to a new rule with an optimal metric distortion of $3$.

\section{Notation}

Throughout the paper, we assume that we are given a set $N = [n]$ of voters and a set $C = \{c_1, \dots, c_m\}$ of candidates. Further, each voter $i \in N$ has a strict and complete \emph{preference list} $\succ_i \subseteq C \times C $ with $\succ = (\succ_i)_{i \in N}$ forming the preference profile. Together, $(N,C, 
\succ)$ form our election instance. For a given profile $\succ$, we let $\succ^R$ be the preference profile such that $c \;\succ^R_i\; c'$ if and only if $c' \succ_i c$, i.e., the reversed preference profile. Further, we write $c_i \succeq c_j$ if and only if $c_i \succ c_j$ or $c_i = c_j$. For a set $N' \subseteq N$ of voters and sets $C', C''\subseteq C$ of candidates we write that $C' \succ_{N'} C''$ if and only if $c' \succ_i c''$ for all $i \in N'$, $c' \in C'$, and $c'' \in C''$. 

Additionally, we need the notion of cloned candidates. For a given  frequency function $f\colon C \to \mathbb N_0$, the instance $(N, C^f, \succ^f)$ is the instance in which every candidate $c$ is replaced by $f(c)$ clones in every preference list. The ranking between the clones is arbitrary but fixed across voters, and candidates with $f(c) = 0$ are deleted. We denote by $\plu(c)$ the plurality score of candidate $c$, by $\boto_i(C')$ the worst ranked candidate $c \in C'$ by voter $i$. Hence, the instance $(N, C^{\plu}, \succ^{\plu})$ denotes the instance in which every candidate is cloned as many times as they appear first in any preference profile. We not that for this special instance $n = m$ holds.

For any $c \in C$ and $N' \subseteq N$ we define \[D_c(N') \coloneqq \{c' \in C \colon c \succeq_i c' \text{ for some } i \in N'\}\] as the set of all candidates which are not ranked better than $c$ by at least one voter in $N'$. Note that, by definition, $c \in D_c(N')$ for all $c \in C$ and non-empty $N' \subseteq N$.

We say that a utility function $u\colon C \to \mathbb{R}^+$ is consistent with a preference list $\succ$ if $c_i \succ c_j$ implies $u(c_i) \le u(c_j)$, i.e., the goal is to select the candidate with the lowest social cost. A utility profile $U = (u_i)_{i \in N}$ is a collection of utility functions such that $u_i$ is consistent with $\succ_i$. For a given utility profile $U$ and candidate $c$ the \emph{social welfare} of $c$ is $\sw(c, U) \coloneqq \sum_{i \in N} u_i(c)$.

We call a utility profile \emph{metric} if there exists a metric space $(N \cup C,d)$ with $u_i(c) = d(i,c)$, i.e., if there is a function $d\colon (N \cup C)^2 \to \mathbb{R}^+$ such that for any $i,j,k \in N \cup C$ it holds that (i) $d(i,j) = 0$ if and only if $i = j$, (ii) $d(i,j) = d(j,i)$, and (iii) $d(i,k) \le d(i,j) + d(j,k)$.

Given a preference profile $\succ$, let $\mathcal{U}(\succ)$ be the set of all metric utility profiles consistent with the preference profile. The \emph{metric distortion} of a given candidate is now the worst-case social welfare ratio of this candidate compared to the social welfare winner, over all consistent metric utility profiles, that is 
\[
\dist(c, \succ) = \sup_{U \in \mathcal{U}(\succ)} \frac{\sw(c, U)}{\min_{c' \in C} \sw(c', U)}.
\]
The metric distortion of a voting rule is the worst-case distortion of any candidate the rule selects over all possible preference profiles.
\section{Plurality Matching and Domination Graphs}
To tackle the problem of designing a voting rule achieving a metric distortion of exactly $3$, \citet*{GHS20a} introduced the Plurality Matching rule, based on the so-called domination graph.
\begin{definition}
Given a candidate $c$, the domination graph $G(c) = (V \cup C, E_c)$ is the bipartite graph consisting of the set of voters on the one side, the set of candidates on the other side, with an edge between voter $i$ and candidate $c'$ if and only if $c \succeq_i c'$.
\end{definition}
A \emph{fractional perfect matching} is a function $M \colon E_c \to \mathbb{R}_{\ge 0}$ with $\sum_{e \in \delta(i)} M(e) = 1$ for all $i \in V$ and $\sum_{e \in \delta(c)} M(e) = \frac{n}{m}$ for all $c \in C$. 
As one of their main contributions, \citet{GHS20a} were able to show that there is always a candidate $c$ whose domination graph has a fractional perfect matching.
\begin{theorem}[\citep{GHS20a}]
For any election, there is a candidate $c$ such that $G(c)$ admits a fractional perfect matching. Such a candidate can be found in polynomial time. \label{thm:mat_dom_ex}
\end{theorem}
Whether $G(c)$ admits a fractional perfect matching can be verified in polynomial time using a generalization of Hall's Theorem. 
\begin{lemma}[\citep{GHS20a}]
Given a candidate $c$, the domination graph $G(c)$ admits a fractional perfect matching if and only if for any $N' \subseteq N$, it holds that $\lvert D_c(N')\rvert \ge \frac{m \lvert N'\rvert}{n}$.
\label{lemma:hall}
\end{lemma}
To relate this to the distortion problem, \citet{GHS20a} further showed that for the special case in which candidates are cloned according to their plurality score, any candidate admitting such a fractional perfect matching has metric distortion of at most $3$.
\begin{theorem}[\citep{GHS20a}]
    Given an instance $(N,C,\succ)$, consider the instance $(N, C^{\plu}, \succ^{\plu})$ in which every candidate $c$ appears as many times as they are ranked first by a voter. Any candidate $c$ whose domination graph $G(c)$ admits a fractional perfect matching in the instance $(N, C^{\plu}, \succ^{\plu})$, has metric distortion of at most $3$ in $(N,C,\succ)$. \label{thm:gk}
\end{theorem}
We call such a candidate a \emph{plurality matching winner}. As such a candidate always exists according to \Cref{thm:mat_dom_ex}, there is also always a candidate with metric distortion of at most $3$. However, while this candidate always exists, selecting it via simply iterating over all candidates and checking whether their domination graph admits such a matching, is somewhat unsatisfying and unnatural when comparing it to other voting rules, especially as it is not obvious (without the involved proof of \citet{GHS20a}), that such a candidate exists.

This motivated \citet{KiKe22a} to introduce a simple rule which can always select a plurality matching winner: The plurality veto rule. 
\begin{algorithm}[t]
\caption{Plurality Veto Rule}
\label{alg:pvr}
Create instance $(N, C^{\plu}, \succ^{\plu})$\;
\For{$i \in [n-1]$}{
$C^{\plu} \gets C^{\plu} \setminus \boto_i(C^{\plu})$ 
}
 return last remaining candidate in $C^{\plu}$\;
\end{algorithm}

We state the rule in Algorithm~\ref{alg:pvr} for the special case of $n = m$ in the instance $(N, C^{\plu}, \succ^{\plu})$.  For the general version, called fractional veto rule, which can also be used to show \Cref{thm:mat_dom_ex}, see the paper by \citet{KiKe22a}.

\section{The Proportional Veto Core}

Unrelated to the previous distortion considerations, \citet{Mou81a, Moul82a} introduced the proportional veto core, which we present here in the formulation of \citet{IaKo21a}. The goal of the proportional veto core is to give groups of voters the right to veto candidates they do not like. The power to veto should be distributed in a proportional fashion: an $\alpha$-fraction of the voters should be able to veto an $\alpha$-fraction of the candidates. To formalize this, for any group of voters $N' \subseteq N$, we say that the \emph{veto power} of $N'$ is $v(N') = \lceil \frac{m \lvert N'\rvert}{n}\rceil - 1$. A candidate $c$ is now considered to be vetoed, if there is some group $N'$ of voters which rank at least $m - v(N')$ candidates all better than $c$, that is, they could veto $v(N')$ candidates together with $c$. 
\begin{definition}
    A candidate $c$ is \emph{vetoed} if there is some set $N' \subseteq N$ of voters and set $C' \subseteq C$ of candidates with $c \notin C'$ such that $\lvert C'\rvert \ge m - v(N')$ and $C' \succ_{N'} C \setminus C'$.
    The \emph{proportional veto core} consists of all alternatives which are not vetoed. 
\end{definition}
The proportional veto core is non-empty \citep{Mou81a, Moul82a} and can be computed using the voting by veto tokens rule introduced by \citet{Mou81a} and further analyzed by \cite{IaKo21a}. This rule works exactly the same as the fractional veto rule of \cite{KiKe22a} and can be computed in polynomial time \citep{IaKo21a}. Additionally, both \citet{Moul83a} and \citet[Proposition~21]{IaKo21a} showed that a candidate is in the proportional veto core if and only if it can be computed by the voting by veto tokens rule. 

Further, \citet{IaKo21a} introduced the ``Vote by consumption'' rule, which also selects from the proportional veto core. Intuitively, their rule works as follows: all voters simultaneously and with the same speed give score to their lowest ranked candidate. Once a candidate reaches a score of $1$ this candidate is eliminated and the voters move on to their next worst candidate. The last candidate remaining wins. As pointed out by \citet{IaKo21a}, this rule is anonymous and can be computed in polynomial time as well.

\section{Proportionality for Solid Coalitions}
Next, we turn to another equivalent formulation of the proportional veto core. For this, we use the well-known notion of \emph{proportionality for solid coalitions (PSC)} \citep{Dumm84a}. It was introduced in the context of proportional committee selection with ranked preferences, as a way to axiomatically measure whether a committee is proportional. In the committee selection problem, instead of selecting a single winner, a set of $k$ winners is selected. Sets $W \subseteq C$ are referred to as committees. 
\begin{definition}
    A committee $W$ of size $k$ satisfies weak Droop proportionality for solid coalitions (weak-PSC) if there is no group $N'$ of voters and set $C'$ of candidates with $\lvert N' \rvert > \frac{\lvert C' \rvert n}{k+1} $ such that $C' \setminus W \neq \emptyset$ and $C' \succ_{N'} C \setminus C'$. \label{def:solid_coalition}
\end{definition}
In a nutshell, if a group of voters agrees on a prefix in their rankings and could afford this prefix, the candidates in the prefix need to be included in the committee. Rules satisfying this notion include Phragmén's sequential rule for ordinal ballots \citep{Jans16a}, the Droop variants of the expanding approvals rule \citep{AzLe20a}, and the well-studied Single Transferable Vote (STV) \citep{Tide95a}.

Interestingly enough, Phragmén's sequential rule for ordinal ballots behaves equivalently to veto-by-consumption run on the reversed preferences profile: the first $k$ candidates discarded by veto-by-consumption on profile $\succ^R$, are exactly the $k$ candidates elected by Phragmén's sequential rule.

This is no coincidence, however, as we show in the next section.

\section{Equivalence Result}
Here we show our main result: For a candidate $c$ it is equivalent that $c$ is selected by the plurality matching rule, is in the proportional veto core, and that the committee $C \setminus \{c\}$ of size $m-1$ satisfies weak-PSC in the reversed preference profile.
\begin{theorem}
    Given an election instance $(N, C, \succ)$ and candidate $c \in C$, the following three statements are equivalent.
    \begin{itemize}
        \item[(i)] The domination graph $G(c)$ of candidate $c$ admits a fractional perfect matching.
        \item[(ii)] Candidate $c$ is in the proportional veto core.
        \item[(iii)] The committee $C \setminus \{c\}$ satisfies weak-PSC in the instance $(N, C, \succ^R)$.
    \end{itemize}
\end{theorem}
\begin{proof}
    (i) $\Rightarrow$ (ii): First, assume that $c$ is not in the proportional veto core. Then there is a group $N' \subseteq N$ of voters and set $C' \subseteq C$ of candidates with $c \notin C'$, $\lvert C'\rvert \ge m - \lceil \frac{m \lvert N'\rvert}{n}\rceil + 1$, and $C' \succ_{N'} \{c\}$. Thus, we know that $\lvert D_c(N')\rvert  \le \lceil \frac{m \lvert N'\rvert}{n}\rceil - 1 < \frac{m \lvert N'\rvert}{n}$ and by \Cref{lemma:hall} the domination graph $G(c)$ does not admit a perfect fractional matching.

    (ii) $\Rightarrow$ (iii): Next, assume that the committee $C \setminus \{c\}$ does not satisfy weak-PSC in the instance $(N, C, \succ^R)$. Thus, there is a set $N'$ of voters and a set $C'$ of candidates with $c \in C'$ such that $C' \succ_{N'} C \setminus C'$ and $\lvert C' \rvert < \frac{\lvert N'\rvert (k+1)}{n} = \frac{\lvert N' \rvert m}{n}$. Since, $\lvert C'\rvert$ is integral, we also get that $\lvert C \setminus C'\rvert = m - \lvert C'\rvert \ge m - \lceil \frac{m \lvert N'\rvert}{n}\rceil + 1$ and hence, the set $C \setminus C'$ together with $N'$ blocks $c$.

    (iii) $\Rightarrow$ (i): Finally, we assume that the domination graph of $c$ does not admit a fractional perfect matching. Then by \Cref{lemma:hall} there is some group $N'$ of voters with $\lvert D_c(N') \rvert < \frac{m \lvert N' \rvert}{n}$, and therefore the candidate set $\{c' \in C \colon c \succ_{N'} c'\} = \{c' \in C \colon c' \;\succ^R_{N'}\; c\} $ together with $N'$ witnesses a weak-PSC violation.
\end{proof}
Thus, all rules for either selecting from the proportional veto core or selecting committees satisfying weak-PSC can also be turned into rules achieving a metric distortion of $3$ by applying them to the instance $(N, C^{\plu}, \succ^{\plu})$. These include the anonymous veto-by-consumption rule (or Phragmén’s sequential rule for ordinal ballots), but also committee selection rules such as STV or EAR,  which satisfy weak-PSC \citep{AzLe20a}.
\section{Matching under Preferences}
Finally, we turn to our fourth setting: Matching under Preferences. More specifically, we investigate the random assignment problem of \citet{BoMo01a}. Like \citet{IaKo21a}, we observe that the Probabilistic Serial rule of \citet{BoMo01a} can be seen as running Phragmén's sequential rule for $k = m$ (or the veto-by-consumption rule for the reversed preference profile) and interpreting the ``eaten part'' as the probabilities. 

 To stay consistent with the notation, we also assume that we are given a set  $N$ of voters, $C$ of candidates and a preference profile $\succ$. Further, we assume that we are given a restriction $k$ on the number of candidates matched. Now, a random assignment $\mu\colon N \times C \to \mathbb{R}_{\ge 0}$ is a function assigning each voter-candidate pair a probability to be matched. That is, $\sum_{i \in N} \mu(i,c) \le 1$ for each $c \in C$. Further, to implement the constraint of at most $k$ candidates being matched, we assume that $\sum_{c \in C} \mu(i,c) = \frac{k}{n}$ for each $i \in N$, that is, each voter is matched with the same probability. We note that this is the same as a price system in the committee selection literature, \citep{PeSk20a, BrPe23a}.
 
 Such a random assignment can also be seen as a lottery over matchings. Given a set of matchings $M_1, \dots, M_\ell$ and probabilities $p_1, \dots, p_\ell \in (0,1]$ with $\sum_{i \in [\ell]} p_i = 1$, we say that this set implements  $\mu$ if $\mu(i,c) = \sum_{j \in [\ell] \colon M(i) = c} p_j$.  The support of $\mu$ are now all matchings $M$ is of size $k$ that are in some set implementing $\mu$. 
 
 Finally, we note that the two popular assignment mechanisms, Random Priority and Probabilistic Serial, can easily be adapted to this setting of matching at most $k$ items: Random Priority by ending the process after $k$ items have been picked, and Probabilistic Serial by only running it for a time of $\frac{k}{n}$. 

Random Priority and Probabilistic Serial are both known to satisfy ex-post efficiency, i.e., the actual allocations they randomize over are Pareto efficient. We formalize this as follows.
\begin{definition}
	A matching $M$ of size $k$ is \emph{Pareto efficient}, if there is no matching $M'$ of size $k$ such that $M'(i) \succeq_i M(i)$ for every $i \in N$ and $M'(i) \succ_i M(i)$ for some  $i \in N$. 
\end{definition}
Further, we say that a random assignment is \emph{ex-post efficient} if every matching in its support is Pareto efficient.

We can show that Pareto optimality gives us another equivalent formulation of the proportional veto core, at least in the case $n = m$. Thus, any ex-post Pareto optimal mechanism can also be used to select from the proportional veto core in this special case. 
\begin{theorem}
Given a voting instance $(N, C, \succ)$ with $\lvert N \rvert = \lvert C \rvert$, a candidate $c$ is in the proportional veto core if and only if in the instance $(N, C, \succ^R)$ there is a Pareto optimal matching of size $m-1$ matching every candidate in $C\setminus\{c\}$.
\end{theorem}
\begin{proof}
	First, if $c$ is not in the proportional veto core, there is a group $N'$ of voters and $C'$ of candidates with $\lvert C'\rvert \ge m - \lceil \frac{m \lvert N'\rvert}{n}\rceil + 1 = n - \lvert N'\rvert + 1$ and $C' \succ_{N'} \{c\}$ and therefore $\{c\} \succ^R C'$. Since the members of $N'$ are matched to at least $\lvert N'\rvert - 1$ items, there must be one voter in $N'$ matched to someone worse than $c$. Hence, this matching was not Pareto optimal, since this voter could have been matched to $c$ instead. 
 
On the other hand, assume that there is no Pareto optimal matching of size $m-1$ matching every candidate in $C\setminus\{c\}$. It is easy to see that this is equivalent to there not being any matching of size $m-1$ which only matches voters to candidates they prefer to $c$ (for a proof see \citep[Theorem~2]{HMSS21a}). Thus, by Hall's Theorem, there must be a set $N' \subseteq N$ of voters with $C' \coloneqq \{c' \in C \colon c' \succ_i c \text{ for some } i \in N'\}$ such that $\lvert C' \rvert \le \lvert N' \rvert -2$. However, then $C'$ and $N'$ witness a violation of $c$ being in the proportional veto core.
\end{proof}
While this applies to the special case relevant for distortion (in which $n = m$) it does not do so in general. For example, consider an instance with $4$ voters and $3$ candidates. Two of the voters have the ranking $c_1 \succ c_2 \succ c_3$ while two other ones have the ranking $c_3 \succ c_2 \succ c_1$. Here, two voters have a veto power of $3 - \lceil \frac{3 \cdot 2}{4} \rceil + 1 = 2$. Thus, neither $c_1$ nor $c_3$ is in the proportional veto core of this instance. However, the matching which matches the first two voters to $c_3$ and $c_2$ is Pareto optimal in the instance $\succ^R$. Hence, we leave a generalization beyond the case of $n = m$ for future work.

Further, we note that by the equivalence between Pareto optimality and the serial dictatorship mechanism \citep{AbSo98a}, the existence of a Pareto optimal matching of size $m-1$ matching every candidate in $C\setminus\{c\}$ is equivalent to there being a permutation $\sigma(1), \dots, \sigma(n)$ of the voters, such that the last voter $\sigma(n)$ picks candidate $c$ in the serial dictatorship mechanism, thus also drawing the connection to matchings of size $m$.

\section{Conclusion and Outlook}
We have shown that the domination graph concept of \cite{GHS20a} and the plurality veto rule of \cite{KiKe22a} are inherently related to three concepts from social choice: The proportional veto core, proportionality for solid coalitions, and Pareto optimal matchings. This gives rise to several more possible algorithms achieving an optimal metric distortion of $3$. Most prominently, Phragmén's ordinal method (or, equivalently, veto-by-consumption or probabilistic serial) can be turned into such a rule, by cloning each candidate according to their plurality score and then picking the last candidate elected by Phragmén's ordinal method for $k = n$ (or, equivalently, the only remaining one for $k = n-1$). 

So far, all methods known to achieve a metric distortion of $3$ rely on this cloning step. Thus, it would be an interesting direction for future research to explore whether there is a voting rule with distortion $3$ which works without this step.

\subsection*{Acknowledgments}
I thank Markus Brill, David Kempe, Fatih Kızılkaya, Dominik Peters, and Alexandros Voudouris for helpful discussions. This research is supported by the Deutsche Forschungsgemeinschaft (DFG) under the grant
BR 4744/2-1 and the Graduiertenkolleg “Facets of Complexity” (GRK 2434).

\bibliographystyle{abbrvnat}
\bibliography{dist, abb, algo}
\end{document}